\documentclass{article}

\usepackage[preprint]{nips_2018}
\usepackage{hyperref}       % hyperlinks
\usepackage{url}            % simple URL typesetting
\usepackage{comment}
\usepackage{textcomp}
\usepackage{booktabs}
\usepackage{microtype}
\usepackage{datetime}
\usepackage[T1]{fontenc}
\usepackage[utf8]{inputenc}
\usepackage{amsmath, amssymb}
\usepackage{amsthm}
\usepackage{color}
\usepackage{float}
\usepackage[numbers]{natbib}
\usepackage{cases}
\usepackage{xparse}
\usepackage{appendix}
\usepackage{caption,subcaption}
\usepackage{algorithm}
\usepackage{wrapfig}
\usepackage{algpseudocode}
\usepackage{verbatim, xspace}
\usepackage{tikz}
\usetikzlibrary{arrows}
\usetikzlibrary{calc}
\usetikzlibrary{decorations.pathreplacing,angles,quotes}

\newtheorem{theorem}{Theorem}
\newtheorem{definition}[theorem]{Definition}
\newtheorem{lemma}[theorem]{Lemma}

\numberwithin{theorem}{section}

\newcommand{\floor}[1]{\left\lfloor #1 \right\rfloor}
\newcommand{\ceil}[1]{\left\lceil #1 \right\rceil}

\newcommand{\eps}{\varepsilon}
\newcommand{\Oh}{\mathcal{O}}

\global\long\def\brq#1{\left[ #1 \right]}
\global\long\def\pr#1{\mathbb{P}\brq{#1}}
\global\long\def\E#1{\mathbb{E}\brq{#1}}

\DeclareMathOperator*{\argmin}{arg\,min}

\newenvironment{full-version}{}{}

\newcommand{\g}[2]{\mathcal{G}(#1,#2)}

 \def\rem#1{{\marginpar{\raggedright\scriptsize #1}}}
\newcommand{\micr}[1]{\rem{\textcolor{blue}{$\bullet$ #1}}}

\newcommand{\ord}{\rho}
\newcommand{\parent}{f_\ord}

\title{Connected Components at Scale via Local Contractions}
\date{}

\author{Jakub Łącki\\
    Google Research\\
    New York, USA\\
    \texttt{jlacki@google.com} \\
    \And 
    Vahab Mirrokni\\
    Google Research\\
    New York, USA\\
    \texttt{mirrokni@google.com} \\
    \And 
     Michał Włodarczyk\\
     University of Warsaw\\
     Warsaw, Poland\\
     \texttt{m.wlodarczyk@mimuw.edu.pl}
}

\begin{document}

\maketitle
\begin{abstract}
As a fundamental tool in hierarchical graph clustering, computing connected components has been a central problem in large-scale data mining. While many known algorithms have been developed for this problem, they are either not scalable in practice or lack strong theoretical guarantees on the parallel running time, that is, the number of communication rounds. So far, the best proven guarantee is $\Oh(\log n)$, which matches the running time in the PRAM model. 

In this paper, we aim to design a distributed algorithm for this problem that works well in theory and practice. In particular, we present a simple algorithm based on contractions and provide a scalable implementation of it in MapReduce. On the theoretical side, in addition to showing $\Oh(\log n)$ convergence for all graphs, we prove an $\Oh(\log \log n)$ parallel running time with high probability for a certain class of random graphs. We work in the MPC model that captures popular parallel computing frameworks, such as MapReduce, Hadoop or Spark.

On the practical side, we show that our algorithm outperforms the state-of-the-art MapReduce algorithms. To confirm its scalability, we report empirical results on graphs with several trillions of edges.
\end{abstract}

\section{Introduction}
Given the popularity and ease-of-use of MapReduce-like frameworks, developing practical algorithms with good theoretical guarantees for basic graph algorithms is of great importance. An algorithm runs efficiently in a MapReduce framework if it executes few rounds, has low communication complexity, and good load-balancing properties. These features are reflected in the Massively Parallel Computation (MPC) model~\cite{mpc-model, mrc-soda, mpc-goodrich}, which provides an abstraction on the popular parallel computing platforms such as MapReduce, Hadoop, or Spark.

Designing efficient MPC algorithms for graph problems has received considerable attention in recent years. The ultimate goal is usually to give an algorithm that runs in $o(\log n)$ rounds, while keeping the space per machine $\Oh(n)$\footnote{We use the standard notation to denote the numbers of vertices and egdes respectively by $n$ and $m$.}~\cite{mpc-matching,mpc-matching2, mpc-vc, mpc-mis}.

%As a fundamental tool in hierarchical graph clustering, computing connected components has been a central problem in large-scale data mining.
While several algorithms and techniques have been developed for computing connected components in a distributed setting, they have shortcomings from theoretical or practical standpoint. For example, some of previously studied practical algorithms satisfy suboptimal theoretical guarantees~\cite{cc-beyond} or lack scalability, requiring each connected component to be small enough to fit in memory of each machine~\cite{hash-to-min}. The best proven guarantee for the number of rounds is $\Oh(\log n)$, which matches the running time in the PRAM model. However, the $\Oh(\log n)$ bound on round complexity is unsatisfactory in relation to practical applications for two reasons:
\begin{itemize}
\item Even a trivial distributed algorithm for finding connected components runs in $\Oh(d)$ rounds, where d is the diameter of a graph (one example is Hash-Min~\cite{hash-to-min}). On the other hand, in the real-world graphs we often have $d \approx \log n$, which makes the $\Oh(\log n)$ bound as good as the trivial $\Oh(d)$ bound.
\item The $\Oh(\log n)$ bound does not explain the practical performance of existing algorithms, as they run in just a few rounds even on graphs with billions of vertices.
\end{itemize}

In this paper, we aim to design scalable distributed algorithms for MPC model with improved theoretical guarantees and practical implementations in a MapReduce-like framework.

\subsection{Our contributions}
We introduce a new simple distributed algorithm for finding connected components, that we call LocalContraction. By taking advantage of its simple definition, we can easily show it converges in $\Oh(\log n)$ rounds and the communication in each round is only $\Oh(m)$. More notably, we can show that for random graphs, a slightly modified version of LocalContraction converges in only $\Oh(\log \log n)$ rounds, even though the diameter of a graph is $\Omega\left({{\log n} \over {\log\log n}}\right)$ with high probability. This gives a theoretical explanation for the very good practical round complexity of the algorithm observed on real-world graphs.
%The second introduced algorithm -- TreeContraction -- uses a different merging strategy, inspired by the clustering algorithm by~\cite{affinity}, and have similar properties.
%Our experiments indicate that LocalContraction is slightly more efficient.

We note that under a popular conjecture~\cite{cycle-conjecture} achieving $o(\log n)$ rounds for general graphs is impossible. The conjecture states that even distinguishing between a graph that is a cycle on $2n$ vertices and a graph consisting of two cycles of lengths $n$ requires $\Omega(\log n)$ rounds, if the space per each machine is $\Oh(n^{1-\epsilon})$ and the total space of all machines is $\Oh(m)$. 

Moreover, we evaluate the algorithm empirically and run it successfully on several large-scale graphs, including a graph with 854 billion vertices and 6.5 trillion edges, the largest graph that has ever been tested by distributed connected components algorithms. In almost all experiments, LocalContraction turns out to be faster than the best known MapReduce algorithm for computing connected components~\cite{cracker}. While the theoretical upper bound on the total communication complexity of our algorithm is $\Oh(m \log n)$, the experiments indicate that for real-world graphs the communication complexity is only $\Oh(m)$.
In our experiments, in each phase of the algorithm the number of edges that are processed decreases at least 10 times.
We believe that this is one of the reasons of the good practical performance of the algorithm.

\subsection{Related work}

As a fundamental problem in large-scale data mining, computing connected components has been studied extensively in the 
literature~\cite{icde-paper, hash-to-min, cc-beyond, cracker, conf/pkdd/SeidlBF12,wsdm18}.
To the best of our knowledge, none of previously studied (practical) algorithms provide a sublogarithmic round complexity (for a general class of graphs). 
In fact, it has been recently shown that under certain complexity assumptions, achieving $\Oh({1\over \epsilon})$ rounds is impossible~\cite{mpc-lb}, if the space per machine is limited to $\Oh(n^{\epsilon})$. 

As for algorithms with $\Oh(\log n)$ number of rounds, algorithms have been developed even in CRCW and CREW  PRAM models~\cite{Shiloach82,krishnamurthy94,kargerNP99:cc}. While CRCW algorithms can be ported in MapReduce to give a $\Oh(\log{n})$ round algorithm in theory~\cite{mrc-soda}, they require the original graph to be persisted and used in each graph iteration.  CREW PRAM algorithms can be simulated in MapReduce using the result of~\cite{mrc-soda}, but they would require $\Oh(n^2)$ communication per MapReduce iteration on a star graph. Previous work in~\cite{hash-to-min} compared the state of the art $\Oh(\log{n})$ time algorithm with a simpler algorithm called Hash-to-Min, and showed that the latter outperforms the former despite the theoretical guarantees. 

Distributed algorithms achieving $o(\log n)$ round complexity either require $n^{1+\Omega(1)}$ space per machine~\cite{filtering} (which is often an unrealistic assumption) or only apply to very special classes of graphs and are very hard to implement in MapReduce~\cite{geometric}.

The best currently known practical MapReduce algorithm for computing connected components is Cracker~\cite{cracker}.
It has been shown to require $\Oh(\log n)$ rounds, but the best known upper bound on the amount of communication is as high as $\Oh(mn / \log n)$, which is very far from what one observes in experiments.

Recently two algorithms for finding connected components in the BSP model~\cite{bsp} have been proposed~\cite{wsdm18, icde-paper}.
The BSP paradigm is used by recent distributed graph processing systems like Pregel~\cite{pregel} and Giraph~\cite{giraph}. While BSP is generally considered more efficient for graph processing, in congested grids, where fault-tolerance against preemptions is more important, MapReduce has certain advantages~\cite{cc-beyond}.
The BSP algorithms can usually be relatively easily implemented in MapReduce, but the fact that MapReduce reshuffles the entire graph between machines in each round means that MapReduce implementations of BSP algorithms do not perform well in practice.

\subsection{Organization of the paper}
We begin with introducing the notation and the theoretical model in Section~\ref{sec:prelim}.
Then we introduce two algorithms based on contraction in Section~\ref{sec:algos}.
Section~\ref{sec:main} comprises main guarantees on the algorithms' behavior for general graphs and  Section~\ref{sec:random} contains the proof of $\Oh(\log\log n)$ convergence for random graphs.
%Most technical derivations have been omitted due to space constraints and they can found in the supplementary material.
In Section~\ref{sec:exp} we present an empirical evaluation of our algorithms on real-world graphs.
Finally, in Section~\ref{sec:lb} we complete the theoretical analysis of algorithms by giving $\Omega(\log n)$ lower bounds on round complexity in general graphs.

\section{Preliminaries}\label{sec:prelim}
Throughout the paper we use $G$ to denote a graph, $V(G)$ to denote its set of vertices, $E(G)$ to denote its edges, $n = |V(G)|$ and $m = |E(G)|$.
For a vertex $v \in V(G)$ we define $N(v)$ to be the set of all neighbors of $v$ in $G$ together with $v$.
For a set $U \subseteq V(G)$ the set $N(U)$ is a union of $(N(v))_{v\in U}$.
Moreover, by $E(X, Y)$ we denote a set of edges that have one endpoint in $X$ and the other endpoint in $Y$.
We define \emph{weakly connected components} of a directed graph $G$ to be connected components of a graph $G'$ obtained from $G$ by making each edge undirected.

%Given an ordering of vertices, that is a one-to-one function $\ord: V(G) \rightarrow [n]$, we define a function $\parent: V(G) \rightarrow V(G),\, \parent(v) = \argmin_{u\in N(v)} \ord(u)$.
%Mapping $\parent$ induces a directed graph on the vertex set $V$ being a collection of disjoint arborescences.
%Let $\repr$ map each vertex to the root of its arborescence.
%Relation $r_\ord$ represents belonging to the same arborescence, which is obviously is a connectedness and equivalence relation.
%Let $C_\ord(v)$ denote the set of vertices $u$ satisfying $r_\ord(v,u)$.

For an equivalence relation $r \subseteq V \times V$ we define the \emph{contraction of $G$ with respect to $r$}, denoted $G / r$, to be a graph obtained from $G$ by merging vertices $G$ that are in relation $r$.
We refer to vertices of $G / r$ as \emph{nodes} to avoid confusion.

\subsection{Massively Parallel Computation Model}

The class MPC($\eps$) is parameterized by a \emph{space exponent} $\eps \in [0,1]$.
Let $N$ denote the data size in bits and $p$ be the number of machines.
At the beginning the data is divided over $\frac{N}{p}$ machines
and in a single round each machine can perform arbitrary computations over its data.
Then the machines exchange information between each other with a restriction that each machine can receive at most $\Oh(\frac{N}{p^{1-\eps}})$ bits in total in a single round.
Therefore the model allows data replication of order $\Oh(p^\eps)$.
The main measure of effectiveness is the number of computation-communication rounds of the algorithm.

There are some differences in defining the model details among various authors.
For example, \cite{mpc-model} allow unlimited computational power of each machine, whereas \cite{affinity} require them to run in polynomial time and space proportional to the local input size.
Our algorithms work in the stronger model with $\eps = 0$ since the time and space complexities are nearly linear with respect to the local input size.

In some of our algorithms we additionally extend the MPC model with a distributed hash table (key-value store)~\cite{bigtable}.
In each round all other machines can send messages of total size $\Oh(n)$ that define the stored key-value pairs.
In the following round, all machines can query the distributed hash table a total of $\Oh(n)$ time, and for each query the value corresponding to a key is returned immediately.

\section{Our algorithms}\label{sec:algos}

We use the term \emph{phase} to denote a logical part of an algorithm, and \emph{round} to denote a single MapReduce computation.
Both algorithms sample a random ordering of vertices in the beginning of each phase.
This is implemented by assigning each vertex a random hash chosen uniformly from $[0, 1]$ -- this induces a random ordering on each connected component independently and we will focus on a single  component in the analysis.
Note that with this approach we can only compare the priorities of vertices and we cannot access their exact positions in the ordering.

\paragraph{LocalContraction}
Each phase starts with sampling a random ordering $\rho : V(G) \rightarrow [n]$.
Then, each vertex $v$ computes a \emph{label} $\ell_\rho(v)$, which is the smallest priority assigned to one of vertices in $N(N(v))$.
Finally, we merge vertices with the same label to a single node and obtain a smaller graph, which is used as an input to the following phase.
Technically speaking, the process of merging is not necessarily a contraction, since two vertices may have the same label even if there is no edge connecting them.
The procedure terminates when each connected component gets reduced to a single node, or, equivalently, the resulting graph has no edges.

\iffalse
\paragraph{Lazy Local Contraction}
This is a modification of Local Contraction with a property that each vertex will be merged with at least one with constant probability.
This can be enforced by resetting the label to itself with probability $\frac{1}{2}$, independently for each vertex.
\fi

\paragraph{TreeContraction}
Similarly to LocalContraction, at the beginning of each phase each vertex receives a random priority.
Let $\parent(v)$ be the vertex with the lowest hash among $N(v) \setminus \{v\}$.
This mapping induces a directed graph $H$.
Assume $r$ is a relation such that two vertices are in relation $r$ if and only if they are in the same weakly connected component of $H$.
In one phase, TreeContraction contracts $G$ with respect to $r$.

Observe that this time each contracted node is guaranteed to originate from a connected subgraph, so we obtain a minor of the original graph.
Again, the following phase processes the contracted graph and we terminate when each connected component gets reduced to a single node.

\begin{lemma}
The contraction step in both algorithms can be implemented in MPC(0) in $\Oh(1)$ rounds.
\end{lemma}
\begin{full-version}
\begin{proof}
We assume that we have already assigned labels to vertices -- for LocalContraction it is straightforward and for TreeContraction it is described in~Theorem~\ref{thm:tree-rounds}.
The vertices are then assigned to machines without replication.
In case a neighborhood of some vertex is too large to fit into single machine's memory, it is divided among several machines and only the label gets replicated.

For each vertex $v$ and edge $uv$ a mapper sends a key-value pair $(u,\ell(v))$, where $\ell(v)$ stands for the label of $v$.
These messages are grouped again by vertices and the label mapping is applied to $u$.
In the end a new edge $(\ell(u),\ell(v))$ is generated and potential duplicates are being removed in a standard way.
\end{proof}
\end{full-version}

\section{Main properties of the algorithms}\label{sec:main}

%We first show that LocalContraction terminates after $\Oh(\log n)$ phases with high probability.

%The proof is based on the fact that each vertex with constant probability has a neighbor, whose hash is among $n/2$ smallest hashes.
%This means that in each phase all vertices with $n/2$ smallest hashes, as well as a constant fraction of other vertices get merged with a 
%vertex with constant probability the label computed in every step is among $n/2$ smallest labels.
%Thus, the expected number of nodes decreases by a constant factor.

In this section we show that both algorithm behave in a sensible way on any input.
For LocalContraction the argument for convergence relies on a fact that each vertex has a big chance of seeing a neighbor with low priority.

\begin{lemma}
LocalContraction terminates after $\Oh(\log n)$ phases with high probability.
\end{lemma}

\begin{full-version}
\begin{proof}
Recall that in each phase of the algorithm, each vertex $v$ computes a label $\ell_\rho(v)$ and the vertices with equal labels are merged.
We show that the expected number of distinct labels computed in each step is at most $\frac{3n}{4}$.

Consider a vertex $w$ satisfying $\rho(w) \ge \frac{n}{2}$.
As long as the component of $w$ has not been reduced into a single node, we can fix a neighbor $u$ of $w$.
The value of $\rho(u)$ is a uniformly random variable from $[1, n] \setminus \{ \rho(w) \}$.
In particular $\pr{\rho(u) < \frac{n}{2}} \ge \frac{1}{2}$.
Let $Z_w$ be a binary random variable indicating whether $\rho(\ell_\rho(w)) < \frac{n}{2}$.
We can see that $\E{Z_w} \ge \frac{1}{2}$ and, by the linearity of expectation, $\E{\sum_{v : \rho(v) \ge \frac{n}{2}} Z_v} \ge \frac{n}{4}$.

We can now bound the number of distinct labels computed in the first phase, which we denote by $X_1$.
Since $\rho(\ell_\rho(v)) \le \rho(v)$, vertices $v$ with $\rho(v) \leq \frac n2$ are assigned labels that are not greater than $\frac n2$.
Moreover, in expectation at least $\frac n4$ vertices with $\rho(v) \leq \frac n2$ are assigned labels not greater than $\frac n2$.
This implies $\E{X_1} \le \frac{3n}{4}$.

Let $X_k$ denote the number of vertices after $k$ phases of the algorithm or $0$ if the algorithm has terminated.
We have $\E{X_k} \le \left(\frac {3} {4}\right)^k n$ and for $k = (1+c)\cdot \log_{\frac{4}{3}} n$ the expected number of vertices is at most $n^{-c}$.
By Markov's inequality we get $\pr{X_k \ge 1} \le n^{-c}$, what finishes the proof.
\iffalse
Observe that for each vertex $v$ its label stays $f_\rho(v)$ and the label of $f_\rho(v)$ is reset to itself with probability $\frac{1}{4}$.
Let random variable $X_v$ equal $\frac{1}{|C(v)|}$, where $C(v)$ is the set of vertices sharing the same label as $v$.
Since $X_v \le \frac{1}{2}$ with probability $\frac{1}{4}$,
we have $\E{X_v} \le \frac{7}{8}$.
Clearly, the size of the merged graph equals $X = \sum_{v \in V} X_v$,
which is at most $\frac{7n}{8}$ in expectation.

Let $T(n)$ denote the expected number of rounds required to process a graph with $n$ vertices.
It satisfies recursive formula $T(n) = 1 + \E{T(X)}$.
We assume inductively $T(k) \le \log_{\frac{4}{3}} k$ for $k < n$ ($X < n$ with probability 1) and, since logarithm is concave, we check 

\[
T(n) = 1 + \E{T(X)} \le 1 + \E{ \log_{\frac{4}{3}} X} \le 1 + \log_{\frac{3}{4}} (\E{X}) \le 1 + \log_{\frac{4}{3}}{\frac{3n}{4}} = \log_{\frac{4}{3}} n.
\]
\fi
\end{proof}
\end{full-version}

Since all operations within each phase of LocalContraction are local, the above lemma immediately implies the following.

\begin{theorem}\label{thm:local-ub}
LocalContraction can be implemented to run in $\Oh(\log n)$ MapReduce rounds with high probability.
\end{theorem}

To bound the number of phases executed by TreeContraction it suffices to observe that all contracted sets contain at least two vertices.

\begin{lemma}\label{lem:tree-phases}
TreeContraction terminates after $\Oh(\log n)$ phases.
\end{lemma}
\begin{full-version}
\begin{proof}
Since $f_\rho(v) \ne v$, each cluster consists of at least two vertices
and the number of vertices drops by at least a factor of two in each phase.
Therefore, the worst case number of contraction phases is $\Oh(\log n)$.
\end{proof}
\end{full-version}

We also need to handle each phase efficiently.
In order to do so, we analyze the structure of paths induced by the mapping $f_\rho$.
In particular, we show that all paths stabilize after relatively few steps.

\begin{lemma}\label{lem:cycle}
Let $\rho : V \rightarrow [n]$ be a one-to-one function.
Let $\parent(v)$ be the vertex with the smallest value of $\rho$ among $N(v) \setminus \{v\}$.
Then, for each $v$ there exists an integer $d(v)$, such that for each $i \geq d(v)$, $\parent^{i}(v) = \parent^{i+2}(v)$.\footnote{Here, $\parent^{i}$ denotes the $i$-th functional power of $\parent$.}
\end{lemma}

\begin{full-version}
\begin{proof}
Since $\parent$ is defined on a finite domain, there exists $d(v)$ such that the sequence $\parent^{d(v)}(v), \parent^{d(v)+1}(v), \parent^{d(v)+2}(v), \ldots$ is periodic.
Let $r$ be the vertex from this sequence with the smallest value of $\rho(r)$.
We show that $\parent(\parent(r)) = r$.
Clearly, $\rho(\parent(\parent(r))) \geq \rho(r)$ from the choice of $r$.
On the other hand, $r$ is a neighbor of $\parent(r)$, so $\rho(\parent(\parent(r))) \leq \rho(r)$.
Hence $\rho(\parent(\parent(r))) = \rho(r)$, which implies $\parent(\parent(r)) = r$.
This implies $\parent^{i}(v) = \parent^{i+2}(v)$ for each $i \geq d(v)$.
\end{proof}
\end{full-version}

\begin{lemma}\label{lem:depth}
Let us define $\rho$, $\parent$ and $d$ as in the statement of Lemma~\ref{lem:cycle}.
Moreover, assume that $\rho$ is chosen uniformly at random from all one-to-one functions mapping $V$ to $[n]$.
Then, $\max_{v\in V} d(v) = \Oh(\log n)$ with high probability.
\end{lemma}
\begin{full-version}
\begin{proof}
Let us define $p(v, i) := \ord(\parent^i(v))$ if $i \leq d(v)$ and $p(v, i) = 0$ for $i > d(v)$.
In particular $p(v, 0) = \rho(v)$.

We now show that for each $v \in G$ we have $\E{p(v, 2)} < p(v, 0)/2$.
If $\parent^2(v) = v$, this is obvious, as $p(v, 2) = 0$.
Thus, let us consider the opposite case.
	\begin{align*}
		\E{\ord(x) \mid x = \parent^2(v), x \neq v} & = \E{\ord(x) \mid x = \argmin_{u\in N(\parent(v)) \setminus \{\parent(v), v\}} \ord(u)}\\
		& = \E{\min\{\ord(u) \mid u \in N(\parent(v)), u \neq \parent(v), u \neq v\}}\\
		& = \E{\min\{\ord(u) \mid u \in N(\parent(v)), u \neq \parent(v), \rho(u) < \rho(v)\}}.
	\end{align*}

Fix $u'$ such that $u' \in N(\parent(v)), u' \neq \parent(v)$ and $\ord(u') < \ord(v)$.
Then $\E{\rho(u')}$ is a uniformly random number from $\{1, \ldots, \rho(v)-1\}$, which implies $\E{\rho(u')} \leq \rho(v)/2$.
	By applying an obvious formula $\E{\min(X_1, X_2, \ldots, X_k)} \leq \E{X_1}$, we get:
	\begin{align*}
		\E{\min\{\ord(u) \mid u \in N(\parent(v)), u \neq \parent(v), \ord(u) < \ord(v)\}} \leq \E{\rho(u')} \leq \ord(v) / 2 = p(v, 0) / 2.
	\end{align*}
By applying the above property, it follows that $\E{p(v, 2\cdot i)} < p(v, 0)\cdot 2^{-i}$ for each $v \in G$ and $i \geq 0$. This in turn implies that $\E{p(v, 2\cdot(2+c)\cdot \log n)} < p(v, 0) \cdot 2^{-(2+c)\cdot\log n} \leq n^{-(1+c)}$.

Hence, $\pr{d(v) \geq 2\cdot (2+c)\cdot\log n} = \pr{p(v,(2+c)\cdot\log n) \geq 1} \leq n^{-(1+c)}$, where the last step follows from Markov's inequality.
By the union bound we get that $\pr{\max d(v) \geq 2\cdot (2+c)\cdot\log n} \leq n^{-c}$.
%Since $\max d(v) \le n$ with probability 1, we obtain $\E{\max d(v)} \le 3\log n + n\cdot n^{-1} = 3\log n + 1$.
\end{proof}
\end{full-version}

\begin{lemma}\label{lem:wccs}
Let us define $\rho$, $\parent$ and $d$ as in the statement of Lemma~\ref{lem:cycle}.
Let $H$ be the directed graph defined by the mapping $\parent$.
	Then, two vertices $u$ and $w$ belong to the same weakly connected component of $H$ if and only if $\{\parent^{d(u)}(u), \parent^{d(u)+1}(u)\} = \{\parent^{d(w)}(w), \parent^{d(w)+1}(w)\}$.
\end{lemma}
\begin{full-version}
\begin{proof}
Clearly, if two vertices do not belong to the same weakly connected components of $H$, the corresponding sets are different.
It remains to show the converse.

By Lemma~\ref{lem:cycle}, each $i \geq d(v)$, $\parent^{i}(v) = \parent^{i+2}(v)$.
This means that $\{\parent^{d(u)}(u), \parent^{d(u)+1}(u)\}$ is the set of all vertices that are appear infinitely many times in the sequence $v, \parent(v), \parent^2(v), \ldots$.

Let $P$ be a simple path in $H$ connecting $u$ and $w$.
First assume that $P$ is a directed path, say, from $u$ to $w$.
Thus, for some $j$, $\parent^j(u) = w$ and the lemma follows.

In the remaining case $P$ is not a directed path.
From the fact that each vertex of $H$ has a single outgoing edge, there exists a vertex $x$ on $P$ such that paths from $u$ to $x$ and from $w$ to $x$ are directed.
This implies that for some $j$ and $l$ $\parent^j(u) = \parent^l(w)$, which completes the proof.
\end{proof}
\end{full-version}

\begin{theorem}\label{thm:tree-rounds}
TreeContraction can be implemented to run in $\Oh(\log n \log \log n)$ MapReduce rounds with high probability.
Moreover, in the MapReduce model with a distributed hash table TreeContraction can be implemented to run in $\Oh(\log n)$ rounds.
\end{theorem}

\begin{proof}
Let us define $\rho$, $\parent$, $d$ and $H$ as in the statement of Lemma~\ref{lem:wccs}.
In order to implement TreeContraction, for each vertex we need to compute a label, such that two vertices have the same label if they belong to the same weakly connected component of $H$.
By Lemma~\ref{lem:wccs} it suffices to compute $\parent^{d(v)}(v)$ for each $v$.
By Lemma~\ref{lem:cycle} this can be computed once we know $\parent^j(v)$ for a single value $j \geq d(v)$.

Having an access to a distributed hash table one can compute $\parent^{d(v)}(v)$ in using $\Oh(d(v))$ queries to the hash table.
Moreover, $d(v) = \Oh(\log n)$ with high probability by Lemma~\ref{lem:depth}.
This means that computing the labels can be done in one round.
After that the graph can be contracted in a constant number of rounds. % by Lemma TODO...

In order to compute representatives without the hash table we take advantage of the pointer jumping technique.
Consider a subroutine in which after the $i$-th step we have a mapping from each vertex $v$ to $\parent^{2^i}(v)$.
In the $i$-step $v$ queries $w = \parent^{2^i}(v)$ for $\parent^{2^i}(w)$ and computes $\parent^{2^{i+1}}(v) = \parent^{2^i}(w) = \parent^{2^i}(\parent^{2^i}(v))$.

The subroutine requires $\log\max_{v\in V} d(v)$ steps, which is $\Oh(\log\log n)$ w.h.p. due to Lemma~\ref{lem:depth}.
\end{proof}

\iffalse
\begin{theorem}\label{thm:tree-diam}
TreeContraction terminates after at most $d$ phases in the worst case.
\end{theorem}
\begin{proof}
Recall that TreeContraction does not reshuffle the ordering but each contracted node inherits the lowest hash of its cluster.
Let $\rho(v_1) = 1$  and $w$ be an arbitrary vertex.
Consider a path $v_1, v_2, \dots, v_k$, where $v_1$ is as defined above, $v_k = w$, and $k \le d$.
Wherever $v_i$ sees the lowest hash in its neighborhood it must gets contracted with the corresponding vertex.
By an inductive argument $v_i$ is contracted to the same node as $v_1$ after at most $i$ phases.
The claim follows.
\end{proof}
\fi

\section{$\Oh(\log \log n)$ round convergence in random graphs}\label{sec:random}

In this section we give a $\Oh(\log\log n)$ upper bound on the number of rounds of LocalContraction on a certain class of random graphs.
We believe that this provides some explanation for the low number of rounds of the algorithm that we observe in practice.

Recall that probability distribution over graphs with $n$ vertices is called $G(n,p)$ is probability of each edge occurrence is $p$ and these events are independent~\cite{gilbert1959random}.
We relax this restrictive model and allow the occurrences of edges to be \emph{at least as likely} as in $G(n,p)$.

\begin{definition}
We say that a probability distribution $D$ over graphs with $n$ vertices belongs to class $\g{n}{p}$ if for all $u,v \in V(G)$ we have $X_{uv} \le X'_{uv}$ almost surely, where $X_{uv}$ is a binary random variable s.t. $X_{uv} = 1$ if and only if $uw \in E(G(n,p))$, and $X'_{uv}$ is the counterpart of $X_{uv}$ for $D$.
\end{definition}

We will use a notation $G \sim \g n p$ to indicate that $G$ is sampled from a distribution from the class $\g n p$.
Note that the $\g n p$ model is far less restrictive than $G(n, p)$, as adding any fixed set of edges to $G \sim \g n p$ yields a graph from $\g n p$.

We will be interested in distributions with $p > \frac{c\cdot\log n}{n}$ for some constant $c$.
Note that a graph sampled from such a distribution is connected with high probability
and for $p\sim \frac {\log n}{n}$ its diameter is asymptotically $\frac{\log n}{\log\log n}$~\cite{chung2001diameter}.
Note that~\cite{cc-beyond} considered a random graph model, but only in the regime when $p \geq n^{-\frac 34}$.
This is a much easier case, as such graphs have constant diameter with high probability.

For the sake of analysis we introduce an additional step in the algorithm.
Note that it can be implemented in $\Oh(1)$ MapReduce rounds.

\paragraph{MergeToLarge step}
This is a step that is executed at the end of each phase of LocalContraction.
It takes as input the contracted graph computed in this phase.
It is parameterized by sequence $(\alpha_i)$ depending on the graph density.
At the end of the $i$-th phase we detect large nodes, that is, those created by merging at least $\alpha_i$ vertices.
For each large node we compute its priority, which is the $\alpha_i$ largest hash of the vertices it contains (using the vertex hashes from phase $i$).
Our goal is to merge each node $v$ with a large node within a two-hop neighborhood of $v$.
If $v$ has at least one large one in distance at most 2, we merge it with the large node of largest priority.

\begin{lemma}\label{lem:ab}
Suppose $G \sim {\g n {\frac{\alpha_n}{n}}}$, where $\alpha_n \ge 4\ln n$.
Let $A = \{v \in V(G) \mid \ord(v) \le \frac{n}{\alpha_n}\},\, B = V(G) \setminus A$.
Then with probability $1 - \Oh(n^{-3})$:
\begin{enumerate}
\item the number of vertices $v \in B$ such that $N(v) \cap A \ne \emptyset$ is at least $\frac{n}{2}$, and
\item for all  $v \in A$ we have $N(v) \cap B \ne \emptyset$.
\end{enumerate}
\end{lemma}

\begin{full-version}
\begin{proof}
To prove the first claim we fix $v\in B$ and bound

\begin{equation}
\pr{\ell_\rho(v) \not\in A} \le \prod_{u \in A} \pr{ uv \not\in E(G)} \le \left(1 - \frac{\alpha_n}{n}\right)^\frac{n}{\alpha_n} < 1/e.
\end{equation}

%Therefore $\E{\sum_{u\in A} |F^{-1}_\ord(u)|} > (1-1/e)\cdot|B| =  (1-1/e)n - o(n)$.
Hence, the expected number of vertices from $B$ having a neighbor in $A$ is at least $(1-1/e)\cdot|B| =  (1-1/e)n - o(n)$.
Since the events of edges occurrences are lower bounded by events from $G(n,p)$ so are events representing having a neighbor is a particular subgraph.
Therefore we can take advantage of Hoeffding's inequality to obtain that the at least $\frac n2$ vertices of $B$ have a neighbor in $A$ with probability at least  $1 - \exp(-\Omega(n))$.

Now we bound the probability of $v \in A$ having no neighbors in $B$:
\[
\left(1 - \frac{\alpha_n}{n}\right)^{n - \frac{n}{\alpha_n}} \le \left(1 - \frac{\alpha_n}{n}\right)^{\frac{n}{\alpha_n}(4\ln n - 1) } \le e/n^4.
\]

Here, the expected number of vertices from $A$ having no neighbor in $B$ is less than $e/n^3$.
Therefore the probability that there is at least one such vertex is most $e/n^3$ by Markov's inequality.
\end{proof}
\end{full-version}

\begin{lemma}\label{lem:bernoulli}
Consider a Bernoulli scheme with $n$ trials and probability of success in each trial at least $p$, where $np \le \frac{1}{2}$.
The probability of at least one success in such a scheme is at least $\frac{np}{2}$.
\end{lemma}
\begin{full-version}
\begin{proof}
We begin with a coupling argument so that we can lower bound the probability of at least one success with an analogous probability for a Bernoulli scheme with a single success chance exactly $p$.
Then we proceed by induction claiming that after $k$ trials the probability of at least one success is between $\frac{kp}{2}$ and $\frac 1 2$ -- let use denote this event as $A_k$.
The induction thesis holds for $k=1$ because $p$ is at most $\frac 1 2$.
Consider the $(k+1)$-th trial.
We have $\pr{A_{k+1}} = \pr{A_{k}} + \pr{\neg A_{k}}\cdot p \ge \frac{kp}{2} + \frac p 2$.
On the other hand $\pr{A_{k+1}} \le (k+1)\cdot p \le np \le \frac 1 2$, what proves the induction thesis.
\end{proof}
\end{full-version}

\begin{lemma}\label{lem:random}
Suppose $G \sim {\g n {\frac{\alpha_n}{n}}}$ where $\alpha_n \ge (4+4c)\cdot\ln n$ and $\alpha_n = o(n^\frac{1}{3})$.
%\footnote{We remark that the constant 4 is chosen for the sake of simplicity and it could be replaced with any constant greater than 1.}.
Let $H$ be the graph created after a single round of LocalContraction with MergeToLarge step with parameter $\frac{\alpha_n}{4}$.
Then with high probability $|V(H)| \le \frac{n}{\alpha_n}$ and $H \sim {\g m {\Omega(\frac{\alpha^2_n}{m})}}$
under condition $|V(H)|  = m$.
\end{lemma}
\begin{proof}
\iffalse
%For now, we consider a simpler variant and assume $\alpha_n \ge \ln n$ and $\eps = 3/4$.
The algorithm consists of two phases, first of which is tree contraction with respect to an arbitrary ordering $\ord$, that takes $\Oh(\log\log n)$ rounds due to Lemma TODO.
During the second phase, we detect nodes in $G / r_\ord$ that have been created by connecting at least $\alpha_n / 4$ vertices of $G$, and contract them with nodes at distance at most 2 that do not satisfy this condition.
This can be done in $\Oh(1)$ rounds due to Lemma TODO.
\fi

Sampling $G$ from distribution of type ${\g n {\frac{\alpha_n}{n}}}$ and then choosing a random ordering is equivalent to fixing the ordering first and then examining existence of each edge independently.
Since we want to claim that the edges in $H$ are distributed (somewhat) independently, in one phase we are going to examine only a subset of edges of $G$.
We refer to the unexamined edges as \emph{tentative} and we do not make any assumptions on their existence.

\iffalse
We start with estimating probability of the event $\parent(v) = u$ assuming $\ord(u) < \ord(v)$.
For this to happen, the edge $uv$ must exist and no edge $xv$ can be present for $\ord(x) < \ord(u)$.
Since these event are independent we have
\begin{equation}
\pr{\parent(v) = u} \ge \pr{uv}\cdot\prod_{x : \ord(x) < \ord(u)} \pr{\neg xv}
\end{equation}
\fi

Let $C_\ord(v)$ denote the equivalence class of vertex $v$, i.e., the set of vertices sharing the same label.
As in the statement of Lemma~\ref{lem:ab} we define $A = \{\ord(v) \le \frac{n}{\alpha_n} : v \in V(G)\}$ and $B = V(G) \setminus A$.
We examine edges from $E(A,B)$ and from now on assume that the conditions from Lemma~\ref{lem:ab} hold.
Let  $B^*$ indicate those vertices from $B$ that have a neighbor in $A$.
We have $\ell_\rho(v) \in A$ for $v\in B^*$, and, by Lemma~\ref{lem:ab}, $|B^*| \geq {n \over 2}$.
Let $R = \{v \in B^* \mid |C_\ord(v) \cap B| \ge {\alpha_n \over 4}\}$ -- we will also refer to these vertices as \emph{red}.
Observe that $|R| \ge \frac n 4$.
Otherwise for at least $n\over 4$ vertices $v\in B^*$ we would have $\ell_\rho(v) \in A$ and $|C_\ord(v) \cap B| < {\alpha_n \over 4}$.
There can be at most $|A| = \frac{n}{\alpha_n}$ such equivalence classes, therefore $|B^* \setminus R| < \frac n 4$,
a contradiction.

The nodes created by contracting red vertices are called red nodes.
A node is red iff at least $\alpha_n$ vertices $v$ with $\rho(v) > \frac n{\alpha_n}$ are merged when it is created.
This is equivalent to the fact that the $\alpha_n$-th largest vertex hash of these vertices is greater than $\frac n{\alpha_n}$.
As a result, from the definition of MergeToLarge step it follows that if a node's cluster is smaller than ${\alpha_n \over 4}$ and there is a red node at distance at most 2, they will get merged together.

In order to argue that all vertices end up sufficiently close to a red one we will show that with high probability $N(N(R)) = V(G)$.
and for this purpose we examine edges from $E(R, B \setminus R)$
(note that they are still tentative).
We have established that $|R| \ge \frac n 4$ therefore the probability of $v \in B \setminus R$ having no neighbors in $R$ is at most
\begin{equation*}
\left(1 - \frac{\alpha_n}{n}\right)^\frac{n}{4} \le e^\frac{-\alpha_n}{4} \le  e^{-(1+c)\cdot\ln n} = n^{-(1+c)}.
\end{equation*}
Thus $\E{|B \setminus N(R)|} \le n^{-c}$ and $\pr{|B \setminus N(R)| \ge 1} \le n^{-c}$ by Markov's inequality,
so with high probability all vertices in $B \setminus R$ have a neighbor in $R$.

It remains to handle vertices in $A$, that is, to show that $A \subseteq N(N(R))$.
By assumption we know that for each $v \in A$ there is $u \in N(v) \cap B$.
If $u$ belongs to $R$, then $v \in N(R)$
and if $u \in N(R)$, then $v \in N(N(R))$.
In the previous paragraph we have proven that
these are the only possible options.

\iffalse
If there is at least one such $u \in B,\, uv \in E(G)$, it might belong to $R,\, N(R) \setminus R$ or $B \setminus N(R)$.
In first two cases $v$ is in distance at most 2 from a red vertex.
By the union bound, the probability of the remaining two cases is at most
$\Oh(\frac{1}{n})$.
Note that $R$ is a sum of equivalence classes of $\rho$, so $G[R\ / \rho$ is a subgraph of $G / \rho$.
We set $U_1 = (N(N(G[R]/ \rho)),\, U_2 = V(G/ \rho) \setminus U_1$ and obtain $\E{|U_2|} = \Oh(1)$.
\fi

All edges in $B \setminus R$ are connected to those in $R$ and those in $R$ are connected to $A$.
This means that all the labels are below $\frac{n}{\alpha_n}$
 so clearly $m = |V(H)|$ is at most $\frac{n}{\alpha_n}$.
In order to estimate to probability of an edge in $H$ being present,
we take advantage of the fact that the edges in $G[ R]$ are still tentative.
Consider a pair of nodes $x,y\in V(H)$.
By the definition of $R$ sets of vertices contracted to $x,y$ contain at least 
$\alpha_n \over 4$ vertices from $B$ each.
We lower bound the probability of $x,y$ being connected by the probability that at least one of these $\alpha^2_n \over 16$ pairs admit an edge.
By Lemma~\ref{lem:bernoulli}, it is $\Omega(\frac{\alpha^3_n}{n}) = \Omega(\frac{\alpha^2_n}{m})$ for $\alpha_n = o(n^\frac{1}{3})$.
There might be also other edges influencing the graph $H$, but in order to show that $H \sim \g m {\Omega(\frac{\alpha^2_n}{m}}$ it suffices to lower bound corresponding random variables with independent ones with the given probability of occurrences.
The independence of the new edge variables follows from the fact that we examine disjoint sets of edges in $G$.
\end{proof}

\begin{theorem}
Suppose each connected component of $G$ is sampled from $\g {n_i} p$, where $n_i \le n$.
Then after $\Oh(\log\log n)$ rounds of LocalContraction with MergeToLarge step each  component will get reduced to a single node with high probability.
\end{theorem}
\begin{proof}
Note that the algorithm never needs to check the actual position of a node in the ordering and only compares it with other nodes.
This allows us to treat each connected component of $G$ separately and
we apply Lemma~\ref{lem:random} to each of them.

Let $\alpha_{n,i}$ indicate the parameter describing the component distribution after the $i$-th round, starting with $\alpha_{n,0} = \alpha_{n} = \Omega(\log _n)$.
We have $\alpha_{n,i+1} = \Omega(\alpha^2_{n,i})$ as long as $\alpha_{n,i} = o(n^{\frac{1}{3}})$ and
by induction we obtain that  $\alpha_{n,i} = \Omega(\log^{2^i}n)$.
%On the other hand, the number of vertices after the $(i+1)$-th is at most $\frac{n}\alpha_{n,i}$.
Hence, after $k = \Omega(\log\log n)$ rounds, we get $\alpha_{n,k} = n^{\Omega(1)}$.
Each following step with high probability decreases the number of vertices by at least a factor of $\alpha_{n, k}$.
Thus, the algorithm terminates in $\Oh(1)$ rounds afterwards.
\end{proof}

\section{Empirical study}\label{sec:exp}

\begin{table}

\caption{\label{fig:graphs}Graphs used in the empirical study.}

\medskip
	\centering
	\begin{tabular}{cccccc}
	 \toprule
  & \bf Nodes & \bf Edges & \bf Largest CC\\
	\midrule
\bf Orkut & 3M & 117M & 3M \\
\bf Friendster & 65M & 1.8B & 65M \\
\bf Clueweb & 955M & 37B & 950M \\
\bf videos & 92B & 626B & 18B \\
\bf webpages & 854B & 6.5T & 7B \\
	 \bottomrule
\end{tabular}
\end{table}

In order to evaluate our algorithms empirically we have used five graphs described in Table~\ref{fig:graphs}.
Orkut and Friendster are social networks from SNAP collection~\cite{SNAP}.
Clueweb~\cite{clueweb1,clueweb2} is publicly available web-crawl graph.
The two non-public data sets are graphs representing pairs of similar entities in large collections of (a) video segments and (b) webpages.
Note that the webpages graph has more edges and over three times more vertices than the largest graph that has been previously used to test distributed connected components algorithms~\cite{wsdm18}.

We compare our algorithms (LocalContraction and TreeContraction) to three previously known algorithms: Cracker~\cite{cracker}, Two-Phase~\cite{cc-beyond} and Hash-To-Min~\cite{hash-to-min}.
These algorithms have been the state-of-the art MapReduce algorithms for computing connected components in recent years.
We implemented all algorithms ourselves in a MapReduce framework.

The implementations of two algorithms, TreeContraction and Two-Phase, additionally use a distributed hash table.
The implementation of Two-Phase that uses a distributed hash table seems to be most efficient  than in the experiments given in~\cite{cc-beyond}.
It allows to execute a sequence of large-star operations followed by a small-star operation in constant number of rounds and thus we count this whole sequence as one phase.

The Cracker algorithm in each phase modifies the edges of the graph and based on the new set of edges deactivates some of the vertices, excluding them from future phases.
We observe that Cracker is equivalent to the following algorithm.
Assume that each node is assigned a random priority.
First, rewire the edges of the graph just as in Hash-To-Min algorithm~\cite{hash-to-min}. Then, compute labels $\ell_p(v) = \min_{w \in N(v)}{\rho(w)}$ and merge together all vertices that have the same label.
This makes it easy to implement the algorithm in a similar way to our algorithms and minimizes the potential differences in the running times caused by using completely different implementations.

For algorithms that contract the graph (LocalContraction, TreeContraction and Cracker) we have introduced two optimizations.
If after some phase the contracted graph is small enough (e.g. has few tens of millions of edges), we send it to one machine that finds its connected components in a single round.
On this machine we use union-find algorithm to find connected components, as it can process incoming edges in a streaming fashion and only use space proportional to the number of vertices.
Note that this optimization does not apply to Two-Phase and Hash-To-Min, which do not modify the set of vertices in the graph.
Moreover, after each phase we can get rid of all isolated nodes from the contracted graph, as their connected component assignment is clear.

\iffalse
The $i$-th phase of each algorithm takes as input an undirected graph $G_i$ and produces a contracted graph $G_{i+1}$ and a mapping $c_i : V(G_i) \rightarrow V(G_{i+1})$ that describes which vertices of $G_i$ got merged into each vertex of $G_{i+1}$.

After some number of phases we obtain a graph $G_k$, whose vertices correspond to the connected components of $G_1$.
The mappings $c_i$ need then to be composed to obtain a mapping from vertices of $G_1$ to vertices of $G_k$ (i.e., the connected components).
\fi

\begin{table}
\caption{\label{fig:phases}Numbers of phases used by each algorithm.}
\medskip\centering

	\begin{tabular}{cccccc}
	\toprule
 & \bf LocalContraction & \bf  TreeContraction & \bf Cracker & \bf Two-Phase & \bf Hash-To-Min\\
	\midrule
\bf Orkut & 2 & 2 & 2 & 3 & 6\\
\bf Friendster & 3 & 3 & 3 & 3 & 8\\
\bf Clueweb & 3 & 3 & 3 & 3 & x\\
		\bf videos & 5 & 4 & 4 & x & x\\
\bf webpages & 5 & 4 & 4 & x & x\\
	\bottomrule
	\end{tabular}
\end{table}

The experimental results are given in Tables~\ref{fig:phases} and~\ref{fig:times}.
When computing running times, we have taken a median from three runs.
The missing entries in both figures correspond to runs that have run out of memory or were timed out.
Due to large resource usage, we only completed one run of Cracker on webpages dataset, and thus we present an approximate result.
%For the two largest graphs we stopped execution after an algorithm run more than three times slower than the best run we observed.

The number of phases in LocalContraction, TreeContraction, and Cracker are very similar and do not exceed five, even for graphs with billions of vertices.
However, Cracker is on average slower than LocalContraction.
We believe that this caused by performing more complicated graph transformations in each step.

\begin{table}

	\caption{\label{fig:times}Relative running times.}

\medskip\centering

 \begin{tabular}{cccccc}
	 \toprule
 & \bf LocalContraction & \bf TreeContraction & \bf Cracker & \bf Two-Phase & \bf Hash-To-Min\\
 \midrule
\bf Orkut & \bf 1.00 & 1.64 & 1.38 & 5.77 & 5.84\\
\bf Friendster & \bf 1.00 & 1.25 & 1.16 & 1.73 & 20.27\\
\bf Clueweb & 1.08 & \bf 1.00 & 2.87 & 1.92 & x\\
\bf videos & 1.03 & 1.08 & \bf 1.00 & x & x\\
\bf webpages & \bf 1.00 & 2.17 & \texttildelow 3 & x & x\\
\bottomrule
\end{tabular}

\end{table}

\begin{figure}
\centering
		\centering
		\subcaptionbox{Friendster}{
			\includegraphics[width=0.45\linewidth]{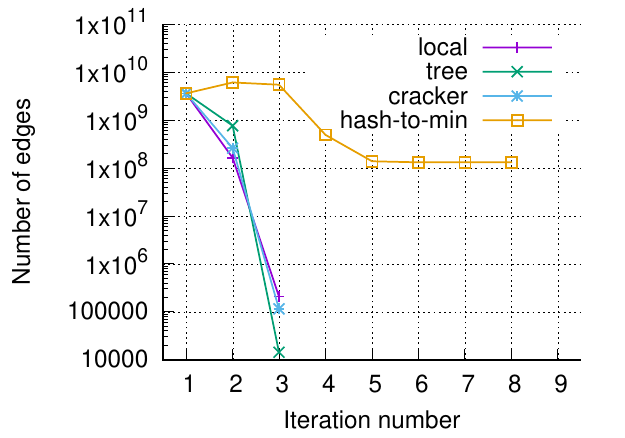}
			}
		\subcaptionbox{videos}{
			\includegraphics[width=0.45\linewidth]{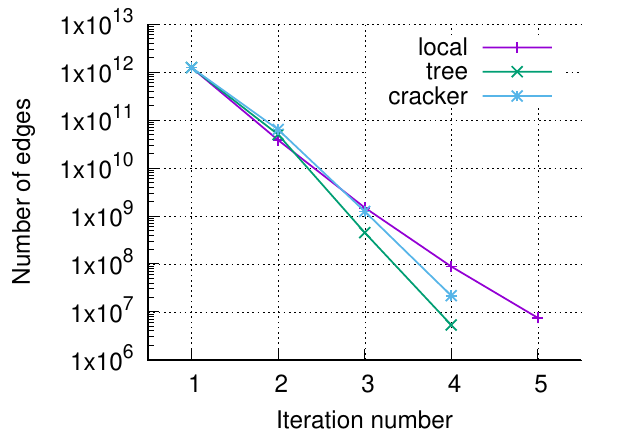}}

	\caption{\label{fig:edges} Numbers of edges at the beginning of each iteration.}
\end{figure}
Finally, Fig.~\ref{fig:edges} shows the effect of graph contraction, by showing the numbers of edges at the beginning of each phase for two datasets (for the other datasets the graphs would look similar). In every dataset and each phase of LocalContraction the number of edges decreases by a factor of at least 10.

\section{Lower bounds}\label{sec:lb}

The Hash-to-Min algorithm has been conjectured to run in $\Oh(\log d)$ rounds~\cite{hash-to-min} but
so far no parallel algorithm has been proved to terminate in $o(\log n)$ or $o(d)$ rounds for all graphs while keeping moderate communication.
One can achieve $\Oh(\log d)$ rounds with the  Hash-to-All algorithm~\cite{hash-to-min}, but it is burdened with a quadratic communication complexity.

In this section we show that none of the considered algorithms can be proved to work in $o(\log n)$ for all graphs.

\begin{theorem}
LocalContraction, Cracker, and Hash-To-Min require $\Omega(\log n)$ rounds on path of length~$n$.
\end{theorem}

\begin{full-version}
\begin{proof}
In a single phase LocalContraction connects vertices at distance at most 4 so it can shorten the path at most 5 times in a phase.
Cracker and Hash-To-Min insert new edges to the graph, however they are constructed by concatenating edges sharing one common endpoint.
By induction one can see that after the $k$-th round the edges of $v$ are at distance at most $2^k$ from $v$.
\end{proof}
\end{full-version}

Lower bound for TreeContraction involves a more sophisticated argument and works only in a randomized setup.
This is inevitable because there is an ordering for which TreeContraction (with a distributed hash table) processes a path in a single round.

\begin{theorem}
TreeContraction requires $\Omega(\log n)$ phases to process a path of length~$n$ with high probability.
\end{theorem}

\begin{full-version}
\begin{proof}
Consider a subpath of length 5.
With probability $\frac 1 5$ the vertex in the middle receives the lowest priority and becomes one of the two vertices that define a contracted weakly connected component.
Let us divide the path into $k = \floor{\frac{n}{5}} \ge \frac n 5 - 1$ segments and let $X_i$ be a binary random variable indicating whether the middle vertex of the $i$-th segment becomes a root.

We define $X = \sum_{i=1}^k X_i$.
Clearly $\E{X_i} = \frac 1 5$ and $\E{X} \ge \frac n {25} - \frac 1 5$.
By Hoeffding's inequality $\pr{X \le \frac n {26}} = \exp(-\Omega(n))$.
Since the class of paths is closed under taking contractions and $X$ is a lower bound on the length of the contracted path,
we can iterate this argument.
Hence, after $\log_{26} n$ rounds, the algorithm will be still running with high probability.
\end{proof}
\end{full-version}

\section{Acknowledgments}
We would like to thank Stefano Leonardi for helpful discussions, in particular on ideas that led to the proof of Theorem~\ref{thm:tree-rounds}.

\bibliographystyle{alpha}
\bibliography{cc}

\end{document}